\documentclass[a4paper]{article}
\usepackage[a4paper,includeheadfoot,margin=2.54cm]{geometry}

\usepackage[T1]{fontenc}
\usepackage{amsthm}
\usepackage{amsmath}
\usepackage{amssymb}
\usepackage[affil-it]{authblk}
\usepackage{microtype}
\usepackage{cite}
\usepackage{soul}
\usepackage{url}
\usepackage{todonotes}
\usepackage{enumitem}

\usepackage{caption}
\usepackage{etoolbox}
\usepackage{lmodern}
\makeatletter
\patchcmd{\@maketitle}{\LARGE \@title}{\fontsize{16}{19.2}\selectfont\@title}{}{}\makeatother

\usepackage{graphicx}
\graphicspath{{./figs/}}
\usepackage{wrapfig}

\theoremstyle{plain}
\newtheorem{theorem}{Theorem}[section]
\newtheorem{lemma}[theorem]{Lemma}

\usepackage[noend]{algpseudocode}
\usepackage{algorithm}
\usepackage{xspace}
\usepackage{xcolor}

\title{Faster Algorithms for Growing Prioritized Disks and 
Rectangles\footnote{A preliminary version appeared as
H.-K.-Ahn, S.~W. Bae, J.~Choi, M.~Korman, W.~Mulzer, E.~Oh, J.-W.~Park,
A.~v.~Renssen, A.~Vigneron. \emph{Faster Algorithms for Growing
Prioritized Disks and Rectangles}. Proc.~28th ISAAC, pp.~3:1--3:13.
The work by H.-K. Ahn, J. Choi, E. Oh was supported 
by the MSIT (Ministry of Science and ICT), Korea, under the SW Starlab 
support program (IITP--2017--0--00905) supervised by the IITP (Institute 
for Information \& communications Technology Promotion.). S.W. Bae was 
supported by Basic Science Research Program
through the National Research Foundation of Korea (NRF) funded by the
Ministry of Education (2015R1D1A1A01057220 and 2018R1D1A1B07042755).
M.~Korman\@ was supported in part by NSF awards CCF-1422311 and KAKENHI No. 17K12635, Japan. 
 W.~Mulzer  was supported in part by DFG Grants MU 3501/1 and
 MU 3501/2 and by ERC StG 757609.
 J.-W.~Park was supported by the NRF Grant 2011-0030044 (SRC-GAIA) 
 funded by the Korea government (MSIP).
 A.~v.~Renssen was supported by JST ERATO Grant Number JPMJER1201, Japan.
    A. Vigneron was supported by the 2016 Research Fund (1.160054.01) of
UNIST (Ulsan National Institute of Science and Technology).}} 
\author[1]{Hee-Kap Ahn}
\author[2]{Sang Won Bae}
\author[1]{Jongmin Choi}
\author[3]{Matias Korman}
\author[4]{Wolfgang Mulzer}
\author[1]{Eunjin Oh}
\author[5]{Ji-won Park}
\author[6]{Andr\'e van Renssen}
\author[7]{Antoine Vigneron}

\affil[1]{Department of Computer Science and Engineering, POSTECH, Republic of Korea\\
\texttt{\{heekap, icothos, jin9082\}@postech.ac.kr}}
\affil[2]{Division of Computer Science
and Engineering, Kyonggi University, Republic of Korea\\
\texttt{swbae@kgu.ac.kr}}
\affil[3]{Department of Computer Science, Tufts University, United States\\
  \texttt{matias.korman@tufts.edu}}
\affil[4]{Institut f\"ur Informatik, Freie Universit\"at Berlin, Germany\\
  \texttt{mulzer@inf.fu-berlin.de}}
\affil[5]{School of Computing, KAIST, Republic of Korea\\
\texttt{wldnjs1727@kaist.ac.kr}}
\affil[6]{School of Computer Science, University of Sydney, Sydney, Australia\\
  \texttt{andre.vanrenssen@sydney.edu.au}}
\affil[7]{School of Electrical and Computer Engineering, UNIST, Republic of Korea\\
\texttt{antoine@unist.ac.kr}}

\date{}
\begin{document}
\maketitle

\newcommand{\IR}{\mathbb{R}}
\newcommand{\eps}{\varepsilon}

\newcommand{\qbox}{b}
\newcommand{\quadtree}{\mathcal{Q}}
\newcommand{\cnp}{\textup{\textsf{CNP}}\xspace}
\newcommand{\cquadtree}{\mathcal{Q}_C}
\newcommand{\etal}{\emph{et al.}\xspace}

\begin{abstract}
Motivated by map labeling, Funke, Krumpe, and 
Storandt~[IWOCA 2016] introduced the following
problem: we are given a sequence of $n$ disks in the 
plane. Initially, all disks have radius $0$, and they 
grow at constant, but possibly different, 
speeds. 
Whenever two disks touch, the one with the higher 
index disappears. The goal is to determine the elimination order,
i.e., the order in which the disks disappear.
We provide the first general subquadratic algorithm for this 
problem. Our solution extends
to other shapes (e.g., rectangles), and 
it works in any fixed dimension. 

We also describe an alternative algorithm that is
based on quadtrees. Its running time is 
$O\big(n \big(\log n + \min \{ \log \Delta, \log \Phi \}\big)\big)$, 
where $\Delta$ is the ratio of the fastest and the slowest growth 
rate and $\Phi$ is the ratio of the largest and the smallest distance 
between two disk centers. This improves the running times of previous
algorithms by Funke, Krumpe, and Storandt~[IWOCA 2016],  
Bahrdt~\emph{et al.}~[ALENEX 2017], and 
Funke and Storandt~[EuroCG 2017].

Finally, we give an $\Omega(n\log n)$ lower bound, 
showing that our quadtree algorithms are almost tight. 
\end{abstract}

\section{Introduction}

Suppose we have a digital map in which certain locations 
are marked with textual labels. As we zoom out,
the visible area increases, while the individual 
features become smaller. To keep readability, the 
labels must grow during the zooming process. At some point, 
the labels will collide, and some of them have to be eliminated
to avoid clutter.
For an efficient implementation, we would like to determine
when and in which order the labels disappear.

This can be formalized as follows:
we are given a sequence $D_1, \dots, D_n$ of
$n$ \emph{growing disks}. 
Each disk $D_i$ starts out as a point 
$p_i \in \IR^2$, and it grows 
with a fixed \emph{growth rate} $v_i > 0$. 
Thus, at any time $t \geq 0$, the disk $D_i$ is centered 
at $p_i$ and has radius $v_it$.
The index $i$ corresponds
to the \emph{priority} (a smaller index means a higher 
priority). When
two disks meet, we eliminate the one with lower 
priority from the arrangement. 
More precisely, for $1 \leq i < j \leq n$, 
let $t(i,j) = |p_ip_j|/ (v_i + v_j)$ be the time when $D_i$ and $D_j$ 
touch. Then, if neither $D_i$ nor $D_j$ has been 
removed before time $t(i,j)$, we eliminate $D_j$ at 
this time, while $D_i$ remains. Our goal is to determine 
the \emph{elimination order}, that is, the instants of time and the
order in which the disks are removed.
In this version, we chose to represent the labels as disks, 
but many other shapes are possible, e.g., to model 
rectangular labels or map icons with more complex boundary shapes. 
Thus, it is desirable to have a solution that is widely applicable 
and adapts easily to small variations in the problem statement.

This problem was introduced by 
Funke, Krumpe, and Storandt~\cite{FunkeKrSt16}.
They observed that a straightforward simulation 
of the growth process with a priority queue runs
in time $O(n^2 \log n)$.  They also gave an algorithm that
takes expected time $O\big(n(\log^6 n + \Delta^2 \log^2 n+ 
\Delta^4\log n)\big)$, where $\Delta=\max_i v_i/\min_j v_j$ is the 
maximum ratio between two growth rates.
Subsequently, Bahrdt~\etal~\cite{Bahrdt17} improved this
to an algorithm that runs 
in worst-case time $O\big(\Delta^2 n(\log n+\Delta^2)\big)$.
This generalizes to growing balls in arbitrary
fixed dimension $d$, with running time 
$O\big(\Delta^d n(\log n+\Delta^d)\big)$. 
Recently, Funke and Storandt~\cite{FunkeSt17}
presented two further parameterized algorithms for the problem. The 
first algorithm runs in time  
$O\big(n \log \Delta(\log n + \Delta^{d-1})\big)$, for arbitrary 
dimension $d$, while the second algorithm is restricted
to the plane and runs in time $O\big(C n \log^{O(1)} n\big)$, where 
$C$ denotes the number of distinct growth rates. 
If we are interested 
in finding only the first pair of touching disks, our problem 
becomes the \emph{weighted closest pair} of the
disk centers.  Formann showed how to compute it in 
$O(n \log n)$ time, which is optimal~\cite{Formann93}.

\begin{table}
\centering
\begin{tabular}{|c|c|c|c|c|c|} \hline
\textbf{Shape} &\textbf{Time} &\textbf{Space} &\textbf{Method} & \textbf{Sec.} \\ 
\hline\hline
Balls/Boxes, $\IR^d$ & $O(dn^{2})$ & $O(n)$ & Priority sort 
&\S\ref{sec:simple} \\ 
\hline
Disks, $\IR^2$ & expected $O\big(n^{\frac{5}{3}+\eps}\big)$ &
$O\big(n^{\frac{5}{3}+\eps}\big)$ &   &  \\ 
\cline{1-3} 
Rectangles, $\IR^2$ & $O(n^{\frac{11}{6}+\eps})$ & 
$O\big(n^{\frac{11}{6}+\eps}\big)$ &  Bucketing & 
\S\ref{sec:bucketing} \\ 
\cline{1-3} 
$\mathcal{S}\mathcal{A}_k$, $k \geq 4$ & 
$O\big(n^{2 - \frac{1}{2k - 2}+\eps}\big)$ & 
$O\big(n^{2-\frac{1}{2k - 2} + \eps}\big)$ & &  \\
 \hline
Cubes, $\IR^d$ & $O\big(n\log^{d + 2} n\big)$& 
$O\big(n\log^{d + 1} n\big)$ &
Orthogonality& \S\ref{sec:cubes} \\ 
\hline \hline
Disks, $\IR^2$ &
$O(n\log \Phi \min  \{\log \Delta, \log \Phi\})$& $O(n \log \Phi)$ & 
Quadtree  & 
\S\ref{sec:uncomp} \\ \hline
Disks, $\IR^2$ &$O(n(\log n + \min\{\log \Delta, \log \Phi\}))$& 
$O(n)$
& Compressed quadtree  &\S\ref{sec:comp}  \\ \hline
\end{tabular}

\caption{Summary of our results.
The $O(dn^2)$-time algorithm in the first row works for growing 
objects of any shape in $\IR^d$ such that the touching time
of any two of them can be computed in $O(d)$ steps.
$\mathcal{SA}_k$ stands for any semialgebraic 
shape that is described with $k$ parameters. $\Phi$ denotes 
the \emph{spread} of the disk centers and
$\Delta=\max_i v_i/\min_j v_j$ is the 
maximum ratio between two growth rates.}
\label{tab:res}
\end{table}

\subparagraph*{Our results.}
We present a simple algorithm that runs in time $O(dn^2)$ 
in any fixed dimension $d$ (Section~\ref{sec:simple}).
For a faster running time, this method can be 
combined with bucketing and an advanced query data structure for 
lower envelopes of algebraic surfaces~\cite{Agarwal97,Koltun04} 
(Section~\ref{sec:bucketing}).
This yields an algorithm with $O(n^{5/3+\eps})$ expected 
time for disks and $O(n^{11/6 + \eps})$ expected time for 
rectangles in two dimensions.
These are the first subquadratic-time algorithms for the problem.
More generally, we show that the elimination sequence of a set of $n$
growing objects of any semi-algebraic shape described with $k$
parameters can be computed in subquadratic time, for any 
fixed $k \geq 4$.  In Section~\ref{sec:cubes}, we 
consider the case of growing squares. These objects are much 
simpler, and we can use ray shooting techniques 
to get a near-linear running time of $O(n\log^{d+2} n)$.

We also consider a completely different approach 
based on quadtrees (Section~\ref{sec:quad}). 
The running time of these algorithms 
also depends on the \emph{spread} $\Phi$ of the disk 
centers (the ratio of the maximum and the minimum 
distance between two disk centers) and the ratio $\Delta$ 
between the fastest and slowest growth rate. 
Table~\ref{tab:res} summarizes our results.
Finally, we give an $\Omega(n \log n)$ lower bound with a 
simple reduction from sorting.
Our method using compressed quadtrees is thus nearly optimal.

\paragraph{Note.} Parallel to our work, 
Castermans~\etal~\cite{cssv-acgs-18} considered a related
problem for growing squares in the plane. Whenever two squares meet, 
they are replaced by a new one located at their weighted center. 
They are also interested in the elimination/replacement 
sequence. Even though our algorithms are slightly faster 
(by polylogarithmic factors) and more general 
(their algorithm can only handle squares),
they are not comparable, since our techniques do not apply in 
the setting where replacements are possible.

\paragraph{Notation.} For any $1 \leq i \leq n$, we denote by 
$t_i$ the time when disk $D_i$ is eliminated. Since 
$D_1$ remains throughout, we set $t_1 = \infty$. We denote by 
$t(i, j) = |p_ip_j|/(v_i+v_j)$ the time at which the 
disks $D_i$ and $D_j$ would touch, supposing
that no other disk has interfered. 
We assume general position in the sense that all times 
$t(i, j)$, for $i \neq j$, are pairwise distinct.

\section{A simple quadratic algorithm}\label{sec:simple}

We present a simple iterative way to determine the elimination 
times $t_i$. This method will be useful in our bucketing scheme.
As noted above, we have $t_1 = \infty$.
We  need a way to determine $t_i$ if 
$t_1, \dots, t_{i - 1}$ are known, for $i \geq 2$.
The next lemma shows how to do this: we just 
need to find the disk of higher priority that first touches $D_i$ and 
is still alive at that point.

\begin{lemma}\label{lem:quadratic}
Let $i \in \{2, \dots, n\}$, and let
\begin{equation}\label{equ:quadratic}
j^* = \textnormal{argmin}_{j = 1, \dots, i-1} 
\{t(i, j) \mid t(i,j) \leq t_j\}.
\end{equation}
Then, the disk $D_i$ is eliminated by the disk $D_{j^*}$, 
and $t_{i} = t(i, j^*)$.
\end{lemma}

\begin{proof}
On the one hand, we have $t_i \leq t(i,j^*)$, because at time
$t(i, j^*)$, the disk $D_i$ would meet the disk $D_{j^*}$ 
that has higher priority and that has not been eliminated yet
(by (\ref{equ:quadratic}), we have 
$t(i, j^*) \leq t_{j^*}$).
On the other hand, we have $t_i \geq t(i, j^*)$, because 
every disk that $D_i$ could meet before time $t(i, j^*)$
either has lower priority or has been eliminated before the encounter.
\end{proof}

\noindent
The condition (\ref{equ:quadratic}) can be implemented with a
straightforward \textbf{for}-loop.
This leads to Algorithm~\ref{alg:first}.

\begin{algorithm}
  \caption{A quadratic time algorithm\label{alg:first}}
  \begin{algorithmic}[1]
    \Function{EliminationOrder}{$p_1, \dots, p_n$, $v_1, \dots, v_n$}
	\State $t_1 \gets \infty$ 
	\For{$i \gets 2, \dots, n$}
		\State $t_i \gets t(i, 1)$
		\For{$j \gets 2,i-1$}
			\If{$t_j\ge t(i, j)$ and $t_i\ge t(i, j)$}
				\State $t_i \gets t(i, j)$
			\EndIf
		\EndFor
	\EndFor
	\State $S \gets \langle D_1,\dots,D_n \rangle$
	\State Sort $S$ using key $t_i$ for each disk $D_i$
	\State \Return $S$
    \EndFunction
    \end{algorithmic}
\end{algorithm}	

\begin{theorem}\label{thm:first}
Algorithm~\ref{alg:first} computes the elimination order of
a set of prioritized disks in $O(n^2)$ time and $O(n)$ space. 
It generalizes to growing objects of any shape in $\IR^d$ such that 
the touching time of any pair of them can be computed in $O(d)$ 
steps, with running time $O(dn^2)$.
\end{theorem}

\begin{proof}
The correctness follows directly from
Lemma~\ref{lem:quadratic}. The running time and space 
analysis is straightforward.
Lemma~\ref{lem:quadratic} is purely combinatorial and requires only 
that the times $t(i,j)$ are well defined. Thus, 
Algorithm~\ref{alg:first} can be generalized by using an 
appropriate $O(d)$-time subroutine for computing $t(i,j)$.  
Then, the claim is immediate.
\end{proof}

In particular, Theorem~\ref{thm:first} shows that the
elimination order for $d$-dimensional balls or rectangles 
can be computed in $O(dn^2)$ time, for any $d \geq 1$.

\section{A subquadratic algorithm using bucketing}\label{sec:bucketing}

We now improve Algorithm~\ref{alg:first} with
bucketing and an appropriate data structure. The main idea 
is as follows: in Algorithm~\ref{alg:first}, we go through the 
disks by decreasing order of priority and determine for each one which
disk eliminates it. This is 
done by examining each disk of higher priority individually, leading
to a quadratic running time. To avoid this, we partition
the disks into \emph{buckets} of size
$m$, where $m$ will be fixed later. For each bucket $B$, we build a 
data structure that can find in sublinear time 
the disk in $B$ that eliminates the current disk. 
Thus, we process the $m$ disks in each bucket with
a single query, at the expense of an additional overhead
for building the data structure. Furthermore, we must compute
the elimination events within each bucket, using
Algorithm~\ref{alg:first}. 

\paragraph{Elimination queries.}
We now describe the query data structure that is used in a single
bucket.
Let $B \subseteq \{1, \dots, n\}$ be a
contiguous set of $m$ indices, and suppose we know 
the elimination time $t_j$ of every disks $D_j$ with $j \in B$. 
In an \emph{elimination query}, we are given a query
index $q > \max B$, and we ask
for the disk $D_{j^*}$ with $j^* \in B$, that eliminates the
query disk $D_q$ (i.e., the first disk in $B$ that meets
$D_q$ and that has not been eliminated yet by \emph{any} other
disk of higher priority). 
The argument from Lemma~\ref{lem:quadratic} shows that we can 
find $j^*$ by the following slight adaptation of (\ref{equ:quadratic}):
\begin{equation}\label{equ:bucket}
j^* = \text{argmin}_{j \in B} \big\{t(q, j) \mid t(q,j) \leq t_j\big\}.
\end{equation}
This leads to a natural interpretation of elimination queries as 
vertical ray shooting with four-dimensional lower envelopes:
a growing disk $D$ corresponds to a point $(x, y, v) \in \IR^3$,
where $(x, y)$ is the center of $D$ and $v$ is the growth rate.
For each $j \in B$, consider the function $f_j: \IR^3 \rightarrow \IR$ 
defined by
\[
  f_j(x, y, v)=
    \begin{cases}
      t\big((x, y, v), j\big), & \text{if } t\big((x, y, v), j\big) < t_j, \\
      \infty, & \text{otherwise,}
    \end{cases}
\]
where $t\big((x,y,v), j\big)$ denotes the time when $D_j$ and the
growing disk given by $(x, y, v)$ touch.
For $q > \max B$, let $(x_q, y_q, v_q) \in \IR^3$ be the
point that represents $D_q$. Then, (\ref{equ:bucket}) tells us that 
the elimination query 
$q$ corresponds to finding the value $\min_{j \in B} f_j(x_q, y_q, v_q)$ 
vertically above $(x_q,y_q,v_q)$ and the index
$\text{argmin}_{j \in B} f_j(x_q, y_q, v_q)$ of the function that attains
it. The pointwise minimum $E(w) = \min_{j \in B} f_j(w)$ is called
the \emph{lower envelope} of the functions $f_j$, and the problem
of determining the function that achieves this minimum for a given
$w \in \IR^3$ is called a \emph{vertical ray shooting query}.
Vertical ray shooting in lower envelopes is a well-studied 
problem in computational geometry~\cite{HalperinSh17}.
In our case, we can apply the following result of 
Agarwal~\etal~\cite{Agarwal97}:
\begin{theorem}[Theorem~3.3 in Agarwal~\etal~\cite{Agarwal97}]
\label{thm:4dlower_env}
Let $\mathcal{F}$ be a given collection of $m$ trivariate, 
possibly partially defined functions, all algebraic of constant 
maximum degree, and whose domains of definition (if they are 
partially defined) are each defined by a constant number of 
algebraic equalities and inequalities of constant maximum 
degree. Then, for any $\eps > 0$, the lower envelope $E_\mathcal{F}$
of $\mathcal{F}$ can be computed in randomized expected time
$O(m^{3+\eps})$, and stored in a data structure of size 
$O(m^{3+\eps})$, so that, given any query point $w \in \IR^3$,
we can compute $E_\mathcal{F}(w)$, as well as the function\textup{(}s\textup{)}
attaining $E_\mathcal{F}$ at $w$ in $O(\log^2 m)$ time.
\end{theorem}
The data structure of Agarwal~\etal~\cite{Agarwal97} works by
first computing the $0$-, $1$-, and $2$-dimensional faces
of the \emph{minimization diagram} for $\mathcal{F}$, i.e.,
the maximally connected regions in $\IR^3$ in which the lower
envelope of $\mathcal{F}$ is achieved by four, three, or two functions
from $\mathcal{F}$. For this, we go over all pairs
$f, g$ of distinct functions in $\mathcal{F}$, and we 
determine the part of the
two-dimensional surface $\Gamma = \{x \in \IR^3 \mid f(x) = g(x)\}$ 
where the lower envelope is achieved by $f$ and $g$. This 
can be done in expected time $O(m^{1+\eps})$ with a randomized 
incremental construction, by inserting the functions from 
$\mathcal{F} \setminus \{f, g\}$ in a random order and keeping track 
of (the trapezoidal decomposition of) the part of $\Gamma$ 
that is not cut off by any function inserted so 
far~\cite[Theorem~2.3]{Agarwal97}.
Once the faces of the minimization diagram are available,
it can be converted into a suitably well-behaved subdivision
of $\IR^3$~\cite[Lemma~3.1]{Agarwal97} to which the point 
location structure of Preparata and Tamassia can be 
applied~\cite{PreparataTa92}. The structure of Preparata and
Tamassia requires $O\big(m^{3+\eps}\big)$ space and preprocessing time, while
a query needs $O\big(\log^2 m\big)$ steps. The result from
Theorem~\ref{thm:4dlower_env} follows. This directly 
translates into the following lemma on elimination queries:

\begin{lemma}\label{lem:agarwal}
Let $B \subseteq \{1, \dots, n\}$ with $|B| = m$.
Then, for any fixed $\eps > 0$, elimination queries for $B$ 
can be answered in $O(\log^2 m)$ time, using space and expected
preprocessing time $O(m^{3+\eps})$.
\end{lemma}

\paragraph{Bucketing.}
With all the tools in place, we can
now describe our subquadratic algorithm.  
We group the disks into $\lceil n/m \rceil$ \emph{buckets}
$B_1, \dots, B_{\lceil n/m\rceil}$, such that
the $k$th bucket $B_k$ contains the disks 
$D_{(k-1)m+1}$, $\dots$, $D_{km}$ (the last bucket might
not be full). There are $O(n/m)$ buckets,
each with at most $m$ disks.
As in Algorithm~\ref{alg:first}, we compute the 
elimination times $t_1,\dots,t_n$, in 
this order. As soon as the elimination times
of all the disks in a bucket $B_k$ have been determined, we construct 
the elimination query data structure for $B_k$. 
By Lemma~\ref{lem:agarwal}, for each bucket, 
this takes $O(m^{3 + \eps})$ expected time and space, 
for a total of $O(nm^{2+\eps})$ expected time and space. 

Now, in order to determine the elimination time $t_i$ of a
disk $D_i$, we must check all the previous buckets 
(as well as the bucket containing $D_i$). We first perform 
elimination queries for all previous buckets,
that is, the buckets $B_k$ with $1 \leq k \leq \lfloor (i-1)/m \rfloor$. 
There are $O(n/m)$ such queries, each taking $O(\log^2 m)$ time. 
Then, we handle the disks 
that are in the same bucket as $D_i$ by inspecting all of them, 
which takes $O(m)$ time. We return the smallest of all the
resulting elimination times. The total time, for $i = 1, \dots, n$, is thus
$O((n^2/m)\log^2 m + nm)$. 
Overall, we obtain an expected running time
of $O(nm^{2+\eps} + (n^2/m)\log^2 m)$.
We balance the terms by setting $m = \lfloor n^{1/3} \rfloor$,
to get an algorithm that takes $O(n^{5/3+\eps})$ space and expected time.

\begin{theorem}\label{thm:bucketing}
The elimination sequence of a set of $n$ growing disks can be 
computed in $O(n^{5/3+\eps})$ space and expected time,
for any fixed $\eps > 0$.
\end{theorem}

\paragraph{Generalizations.}
The subquadratic algorithm generalizes to other shapes. 
For example, consider the problem of growing rectangles in $\IR^2$. 
Each rectangle is given by
the $x$- and $y$-coordinates of two opposing corners at time $t = 1$
(this lets us 
deduce the center and the speed of the rectangle). These are four 
parameters, so
elimination queries now can be handled by vertical ray shooting
in lower envelopes in $\IR^5$.
In this setting, we  employ a general data structure for
point location in arrangements of high-dimensional surfaces or surface
patches. Chazelle~\etal~\cite{ChazelleEdGuSh91} presented such
a structure, based on geometric divide and conquer.
Their result 
requires a decomposition of (subsets of) this arrangement 
into \emph{elementary cells} of constant complexity and depends
on the complexity of such a decomposition. The best relevant
bounds are due to Koltun~\cite{Koltun04}. By plugging
his combinatorial bounds into the scheme of 
Chazelle~\etal~\cite{ChazelleEdGuSh91}, he obtains the following
theorem:
\begin{theorem}[Theorem~5.1 in~\cite{Koltun04}]\label{thm:koltun}
A collection $\mathcal{F}$ of $m$ fixed-degree algebraic surfaces
or surface patches in $\IR^d$, for $d \geq 4$, can be preprocessed in
time $O(m^{2d - 4 + \eps})$ into a data structure of size 
$O(m^{2d - 4 + \eps})$ such that the cell in the arrangement of
$\mathcal{F}$ that contains a query point can be located in time 
$O(\log m)$.
\end{theorem}
A few remarks are in order: 
the bound in Theorem~\ref{thm:koltun} holds for any fixed $\eps > 0$.
The preprocessing algorithm is deterministic (for this, we need 
to replace the random sampling step in the scheme of 
Chazelle~\etal~\cite{ChazelleEdGuSh91} with the deterministic 
$\eps$-net construction by Chazelle and 
Matou\v{s}ek~\cite[Theorem~4.6]{Chazelle00}, which can be done 
in a black box fashion).
Since the point location actually takes
place in the vertical decomposition of the arrangement of
$\mathcal{F}$, Theorem~\ref{thm:koltun}
can also be used for vertical ray shooting. Thus, 
for our case of $m$ growing rectangles, Theorem~\ref{thm:koltun} means that 
a data structure for elimination queries with query 
time $O(\log m)$ can be constructed in $O(m^{6+\eps})$ space and 
time, for any fixed $\eps > 0$.

We now apply the same approach as for growing disks. With bucket
size $m$, we then obtain an algorithm that determines the 
elimination order in $O(nm^{5+\eps} + (n^2/m)\log m)$ time.
Setting $m = \lfloor n^{1/6} \rfloor$ to balance the terms,
we get the following result:

\begin{theorem}\label{thm:bucketing2}
The elimination sequence of a set of $n$ growing rectangles can 
be computed in $O\big(n^{11/6+\eps}\big)$ deterministic time and space, 
for any $\eps > 0$.
\end{theorem}

More generally, we can use Theorem~\ref{thm:koltun} to handle 
elimination queries for regions defined by any semi-algebraic
shape of constant complexity. If the shape of the object is described 
with $k \geq 4$ parameters, elimination queries translate to vertical ray
shooting in the lower envelope of $m$ 
surfaces or surface patches in $\IR^{k+1}$.
The bucketing approach and Theorem~\ref{thm:koltun} then
yield an algorithm that needs
$O(nm^{2k - 3 +\eps} + (n^2/m)\log n)$ time and
space.
We balance the terms with
$m = \lfloor n^{1/(2k-2)} \rfloor$. This gives an overall running time of 
$O\left(n^{\frac{4k-5}{2k-2}+\eps}\right)$, which is subquadratic for 
any fixed $k \geq 4$.
\begin{theorem}\label{thm:bucketing3}
The elimination sequence of a set of $n$ growing objects of any 
semi-algebraic shape, each described with $k \geq 4$ parameters can 
be computed in $O\left(n^{2-\frac{1}{2k-2}+\eps}\right)$ deterministic 
time and space, for any $\eps > 0$.
\end{theorem}

\section{Growing cubes}\label{sec:cubes}

Axis-aligned cubes in $\IR^d$ are given by $d + 1$ parameters. 
Thus, the general approach from Section~\ref{sec:bucketing} applies.
However, for axis-aligned cubes, much better data structures for
orthogonal range searching and for planar ray-shooting 
can be leveraged for elimination queries. 
In this section, we combine bucketing with orthogonal range searching
and ray shooting techniques to achieve an almost linear bound. 

To simplify the presentation, we first focus on the case $d = 2$. 
A sequence of $n$ growing squares is given by the centers 
$p_1, \dots, p_n$ and the growth rates $v_1, \dots, v_n$. 
At time $t \geq 0$, each square $D_i$ 
has edge length $2v_it$. 
We consider the four \emph{quadrants} around each center 
$p_i = (x_i, y_i)$. The \emph{north}, \emph{east}, \emph{south}, and 
\emph{west} quadrants are, respectively,
$\big\{(x,y) \in \IR^2 \mid y-y_i \geq |x-x_i|\big\}$,   
$\big\{(x,y) \in \IR^2 \mid x-x_i \geq |y-y_i|\big\}$,
$\big\{(x,y) \in \IR^2 \mid -(y-y_i) \geq |x-x_i|\big\}$, 
and $\big\{(x,y) \in \IR^2 \mid -(x-x_i) \geq |y-y_i|\big\}$,
see Figure~\ref{fig:quadrants}.
\begin{figure}
  \centering
  \includegraphics{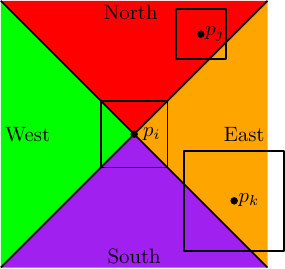}
  \caption{The four quadrants for a growing square $D_i$ with 
  center $p_i$. If the center $p_j$ of $D_j$ lies in the north 
  quadrant, then the possible elimination time between $D_j$ and 
  $D_i$ is determined by the $y$-coordinates (and similarly for 
  $D_k$).}
  \label{fig:quadrants}
\end{figure}

The quadrants determine which sides of the growing square 
$D_i$ are relevant for the elimination time. For example,
suppose that the center $p_j$  of the disk $D_j$ 
is in the north quadrant of $p_i$. 
Then, $D_i$ and $D_j$ will meet when the top side of $D_i$ touches 
the bottom side of $D_j$, and the possible elimination time of $D_i$ 
and $D_j$ is $t(i, j) = (y_j -y_i)/(v_i+v_j)$. 
This can be used for more efficient elimination queries 
as follows: suppose we have a set $B \subset \{1, \dots, n\}$ of 
$m$ growing cubes, and let $q > \max B$ such
that all centers $p_j$ with $j \in B$ lie in the
north quadrant of $p_q$.
Then, an elimination query for $q$ in $B$
just depends on the $y$-coordinates of the disk centers
and the growth rates, the
$x$-coordinates become irrelevant. Thus, we can 
reduce these elimination queries to two-dimensional ray-shooting.

\begin{lemma}\label{lem:cube1}
Let $B \subseteq \{1, \dots, n\}$, $|B|=m$.
We can preprocess $B$ in $O(m \log m)$ time into a data structure 
of $O(m)$ space, so 
that elimination queries can be answered with $O(\log m)$ time, given that
the centers of the squares in $B$ lie 
in the north quadrant of the query square $D_q$.
\end{lemma}

\begin{proof}
We equip the plane with a coordinate system in which the 
horizontal direction is labeled $t$ and the vertical 
direction is labeled $y$.
For each $j \in B$, consider the line segment 
$f_j: t \mapsto y_j-v_jt$, 
defined for $t \in [0, t_j]$; see Figure~\ref{fig:square-queries}.
Let $E(t) = \min_{j \in B} f_j(t)$ be the lower envelope 
of the line segments. An \emph{edge} of $E$ is a maximal contiguous 
interval where the minimum is achieved by a single function $f_j$.
All the line segments $f_j$ begin on the $y$-axis. This means 
that for any pair $f_i$, $f_j$ of line segments, $f_i$ and $f_j$ 
can alternate at most twice on $E$ (i.e., there can be an edge 
from $f_j$, later an edge from $f_i$, and again later an edge 
from $f_j$, but after that $f_i$ cannot appear again).
Combinatorially, this corresponds to a \emph{Davenport-Schinzel 
sequence} of order $2$ with alphabet size $m$, i.e., a sequence 
$\sigma$ of symbols from a finite set $\Sigma$ with $m$ elements 
such that for any two symbols $a, b \in \Sigma$, the pattern 
$a \dots b \dots a \dots b$ does not appear in 
$\sigma$~\cite{SharirAg95}.
It is well known that such a sequence $\sigma$ has length at most 
$2m - 1 = O(m)$~\cite[Chapter~7.1]{Matousek}.

An elimination query for a square $D_q$ with
center $(x_q, y_q)$ and growth rate $v_q$ translates to shooting a 
ray $r$ from the point $(0, y_q)$ with slope $v_q$. Since $r$ and 
all line segments $f_j$ originate from the $y$-axis, and since 
$r$ starts below the lower envelope $E$, the point where 
$r$ meets the first line segment $f_j$ must lie on $E$.
Thus, we first compute $E$
in $O(m \log m)$ time~\cite{Hershberger89},  and then
we build a ray-shooting data structure
for $E$. The latter can be done with $O(m)$ preprocessing 
time and space, and a query time of $O(\log m)$~\cite{Chazelle94}.
The result follows.
\end{proof}
\begin{figure}
\centering
\includegraphics{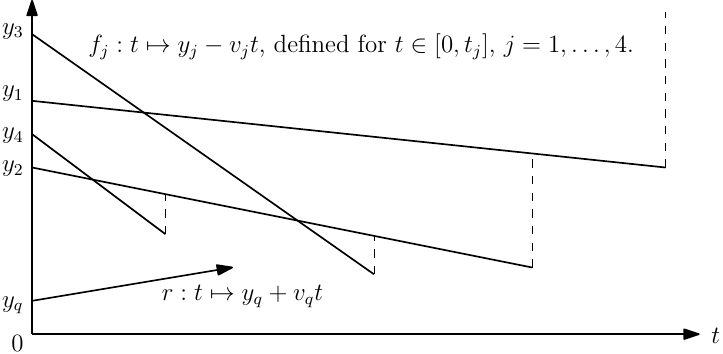}
\caption{The lower envelope of four line segments. An elimination query
  for a square $D_q$ with center $(x_q, y_q)$ and growth rate 
  $v_q$ consists of shooting a 
ray $t \mapsto y_q+v_qt$ from below.}
\label{fig:square-queries}
\end{figure}

We now show how to handle more general elimination queries where 
we do not require $B$ to be in the north quadrant of $D_q$.
This is done using methods from orthogonal range searching.

\begin{lemma}\label{lem:cube2}
Let $B \subseteq \{1, \dots, n\}$, $|B|=m$.
We can preprocess $B$ in time $O(m \log^3 m)$ into a data structure 
of $O(m \log^2 m)$ space, so that 
elimination queries can be answered in $O(\log^3 m)$ time.
\end{lemma}

\begin{figure}
\centering
\includegraphics{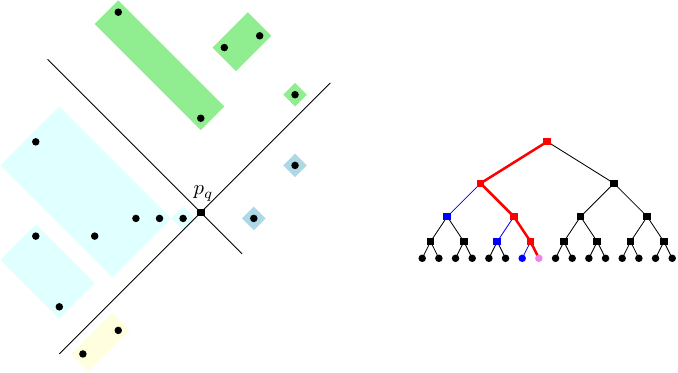}
\caption{(left) A point set $B$ and
the canonical sets for a query point 
$p_q$; (right) the bucketing scheme: each node in the
binary tree corresponds to a bucket. The current square 
is $D_8$ (the eighth leaf from the left). To process $D_8$,
we query all buckets that are left children of the nodes
along the path from the root to the leaf for $D_8$.}
\label{fig:datastructures}
\end{figure}
\begin{proof}
For any square $D_q$, we must be able to perform elimination 
queries on $B$ for all four quadrants that are defined by 
the center $p_q$ of $D_q$.

This is done as follows: let $P$ be the centers of the squares 
in $B$. We build a two-dimensional range tree for 
$B$~\cite[Chapter~5.3]{4M}, 
where the coordinate axes have been rotated by an angle of $\pi/4$.
More precisely, let the rotated coordinates be 
$z_1$ and $z_2$. We sort $P$ according to the $z_1$ coordinate, 
and we build a perfect binary tree $T$ on $P$ for this order, such
that the leaves of $T$ correspond to the points in $P$.
The height of $T$ is $O(\log m)$, and each node $\nu$ of $T$ 
corresponds to a subset $P_\nu$ of $P$, namely the leaves in the
subtree under $\nu$. Each point of $P$ appears 
in $O(\log m)$ subsets. Next, we determine the $z_2$-order 
in each subset $P_\nu$, and we build a binary tree $T_\nu$ on each 
$P_\nu$ for that order, as before. Again, for every node $\nu$ of $T$, 
each node $\mu$ of $T_\nu$ corresponds to a subset $P_{\nu\mu}$ 
of $P_\nu$. We call the subsets $P_{\nu\mu}$ the \emph{canonical sets},
see Figure~\ref{fig:datastructures}(left) for an example. 
Each point in $P_\nu$ appears in $O(\log m)$ subsets 
$P_{\nu\mu}$, so the total size of all canonical sets is
$O(m \log^2 m)$. For each canonical set $P_{\nu\mu}$, we build 
four elimination query structures as in Lemma~\ref{lem:cube1},
one for each quadrant. The total preprocessing time 
is $O(m \log^3 m)$, and the total space requirement is 
$O(m \log^2 m)$.

Now, to process a query $q$, we determine the four quadrants 
of $p_q$, and for each quadrant $Q$, we find the canonical
sets that constitute a partition of $P \cap Q$.
For this, we locate the $z_1$-coordinate of the vertical boundary
of $Q$ in $T$, and we take all the nodes that are left or right
children of the nodes on this search path, depending on whether $Q$
lies to the left or to the right of $p_q$ in the rotated coordinate
system. Then, we perform an 
analogous query in the second level tree for each such node, using
the $z_2$-coordinate of the horizontal boundary of $Q$.
This gives $O(\log^2 m)$ canonical sets that constitute 
a partition of $P \cap Q$. We do an elimination query
for each canonical set, and we return the element that gives
the minimum elimination time for $D_q$. We repeat this for all
four quadrants, and we find the element with
the minimum overall elimination time. Since we query 
$O(\log^2 m)$ canonical sets in total, this takes
$O(\log^3 m)$ time.
\end{proof}

As in Section~\ref{sec:bucketing}, we now apply 
Lemma~\ref{lem:cube2} together with the bucketing
technique. This time, however, the bucketing is done 
in a slightly different way:
we construct a perfect binary tree $T$ whose leaves represent
the squares $D_1, \dots, D_n$, in that order.
A node $\nu \in T$ represents the subset $B_\nu$ of disks that
consists of the leaves in the subtree rooted in $\nu$.

As soon as the elimination
times of all the disks $B_\nu$ associated with a node
$\nu$ of $T$ have been determined, we compute 
the elimination query structure from Lemma~\ref{lem:cube2} for 
$B_\nu$. Thus, when processing a disk $D_i$, the disks
$D_j$, $j < i$, can be partitioned into $O(\log n)$ 
buckets for which elimination query structures have been
constructed (at most one node per level in the tree), see
Figure~\ref{fig:datastructures}(right). Hence, 
we can find $t_i$ in $O(\log^4 n)$ time by querying all
these structures.
Since each disk apperas in $O(\log n)$ structures, 
the total preprocessing time is $O(n \log^4 n)$ and the
total space requirement is $O(n \log^3 n)$.
In higher dimensions, these bounds increase by a factor $O(\log n)$ 
per dimension, as we need one 
more level in the range tree. The following theorem summarizes our
result.

\begin{theorem}\label{thm:cubes}
The elimination sequence of a set of $n$ axis-aligned cubes in 
fixed dimension $d=O(1)$ can be computed in $O(n \log^{d + 2} n)$ time,
using $O(n \log^{d + 1} n)$ space.
\end{theorem}

\section{Quadtree-based approach}\label{sec:quad}
  
Algorithm~\ref{alg:first} can be improved further by noticing that 
in order to find out when a disk is eliminated, it suffices to check
only disks that are \emph{nearby}. Thus, we need a suitable data 
structure to maintain the proximity relations between the disks.
For this, we use a \emph{quadtree}. The nodes in the quadtree
allow us to approximate the disks by square cells and to efficiently
maintain subsets of nearby disks during the growth process. The main
drawback is that the complexity of the algorithm will depend on
the structure of the point set and not just on the number of points.
By compressing the quadtree, we can reduce this dependency, but not
completely eliminate it.
More precisely, let $\Phi$ denote the \emph{spread} of the disk centers and
$\Delta$ denote the ratio of the growth rates, i.e., 
\[
\Phi = \frac{\max_{1 \leq i < j \leq n}|p_ip_j|}
{ \min_{1 \leq i < j \leq n}|p_ip_j|}
\]
and 
\[
\Delta = \frac{\max_{i \in \{1, \dots, n\}} v_i}
{ \min_{j \in \{1, \dots, n\}} v_j}.
\]
We provide two algorithms:
the first algorithm uses a quadtree and runs in 
$O(n\log\Phi\min\{\log\Phi,\log \Delta\})$ time 
and $O(n \log \Phi)$ space.
The second  algorithm uses a compressed quadtree, and
it needs $O(n(\log n + \min\{\log\Phi, \log \Delta\}))$ time
and $O(n)$ space.
To simplify the notation, we 
set $\alpha=\min\{\log\Phi,\log \Delta\}$.

\subsection{A simple quadtree-based algorithm}\label{sec:uncomp}

\begin{figure}
\centering
\includegraphics[scale=1.2]{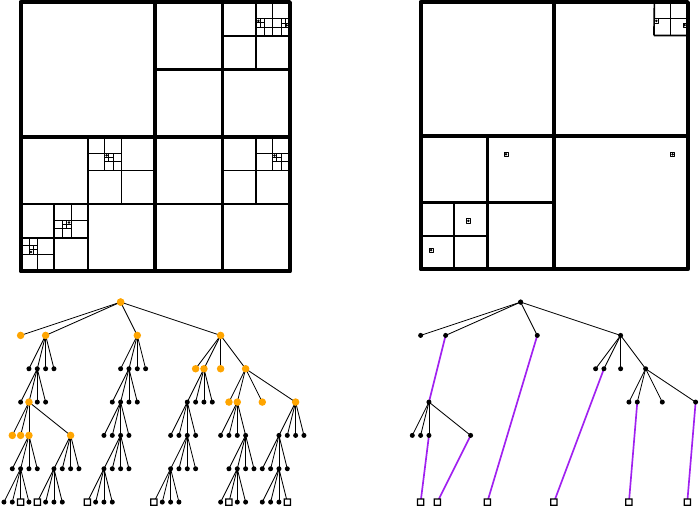}
\caption{A quadtree and its compressed version: (left) 
a quadtree for $6$ disk centers, where the subdivision process
stops once a cell contains at most one disk center and the diameter
of the cell becomes smaller than a quarter of
the smallest distance between disk centers. The nodes that
appear also in the compressed quadtree are marked; (right) the
compressed quadtree.
The compressed edges are shown in bold purple.}
\label{fig:cquadtree}
\end{figure}
We give a simple quadtree-based algorithm to compute the 
elimination sequence for a set of growing disks in the plane.
For this, we first construct a modified quadtree $\quadtree$ for the 
disk centers: suppose that the
maximum distance between two disk centers is $1$ and the minimum
distance is $1/\Phi$. Then, $\quadtree$ will have depth 
$O(\log \Phi)$, and each disk center will be in a leaf cell of 
diameter $O(1/\Phi)$. The growth process of a disk $D_i$ can 
be approximated by tracing the quadtree cells from the leaf cell
that contains $p_i$ to the root. It turns out that if two disks 
collide, then their respective cell approximations must be 
``close'' in the quadtree and of approximately the same size
(roughly up to a factor of $\Delta$). This means that for 
each approximate disk (i.e., quadtree cell), there are 
$O(\alpha)$ other approximate disks that could potentially eliminate it.
Thus, we reduce the number of disk pairs that need to 
be checked in the main loop of Algorithm~\ref{alg:first} from $O(n^2)$ to 
$O(n \alpha \log \Phi)$. 
Details follow.

Without loss of generality, we assume that all disk centers lie 
in the unit square $[0,1]^2$, and that the diameter of this point
set is $1$.  
We construct a \emph{quadtree} $\quadtree$ for the disk 
centers~\cite{MitchellMu17}, with a slight modification.
The quadtree is a rooted tree in which every internal node has four
children. Each node $\nu$ of $\quadtree$ has an associated square
\emph{cell} $b(\nu)$.
To obtain $\quadtree$, we recursively split the unit square. 
In each step, the current node $\nu$
is partitioned into four congruent quadrants (cells) if its corresponding
cell $b(\nu)$ contains one or more disk centers. 
We stop when each cell at the bottom level contains at most one disk 
center and the diameter of the cell becomes smaller than a quarter of
the smallest distance between two disk centers.
Note that in our quadtree, all leaf cells have the same size and
are small in relation to the smallest distance between two disk
centers; see Figure~\ref{fig:cquadtree}.
The quadtree $\quadtree$ has depth $O(\log \Phi)$, and it 
can be constructed in $O(n \log \Phi)$ time and
space. 

We introduce some notation.
For a node $\nu \in \quadtree$, we let
$p(\nu)$ be the parent node of $\nu$.
We denote by
$|\nu|$ the diameter of the cell $b(\nu)$.
For two nodes $\nu, \mu \in \quadtree$, we write $d(\nu, \mu)$
for the smallest distance between a point in $b(\nu)$ and 
a point in $b(\mu)$.
For a point $q$ and a node $\nu \in \quadtree$, we write $d(q, \nu)$
for the smallest distance between $q$ and a point in $b(\nu)$.
The next definitions show what it means to \emph{approximate}
a disk by a quadtree cell.  For $t \geq 0$, 
we let $D_i^t$ be the disk $D_i$ at time $t$.
We say that $D_i^t$ \emph{occupies} a node $\nu$ if 
\begin{enumerate}[label=(\roman*)]
\item the disk center $p_i$ lies in the cell for $\nu$, i.e., 
$p_i \in b(\nu)$;
\item $\nu$ is a leaf or $D_i^t$ covers the whole cell for $\nu$, i.e., 
$b(\nu) \subseteq D_i^t$; and
\item $D_i^t$ has not been eliminated before time $t$.
\end{enumerate}
We denote by $\nu(i, t)$ the node of the
largest cell of $\quadtree$ that is occupied by $D_i^t$.
We may think of $\nu(i, t)$ as an approximate representation
in $\quadtree$ of disk $D_i$ at time $t$. The next lemma shows
that if two disks meet, then their respective approximate
representations must be close and of approximately the same
size.

\begin{lemma}\label{lem:quad.def}
Let $i \in \{2, \dots, n\}$, and let $D_j$, $j \in \{1, \dots, i - 1\})$
be the disk that eliminates $D_i$, i.e., $t_i = t(i,j)$.
Then,
\[ 
d\left(\nu(i, t_i),\nu(j, t_i)\right) \leq 
  2\left(\left|\nu(i, t_i)\right| + 
  \left|\nu(j, t_i)\right|\right), 
\] 
and
\[ 
\frac{1}{4\Delta} \leq \frac{\left|\nu(i, t_i)\right|}{
\left|\nu(j, t_i)\right|} \leq 4\Delta. 
\]
\end{lemma}

\begin{proof}
We state three simple facts from the construction of $\quadtree$
and from the definition of $\nu(\cdot, \cdot)$:
\begin{enumerate}[label=(\roman*)]
\item all non-empty leaf cells have the
same diameter (this is how we construct $\quadtree$); 
\item for any $k \in \{1, \dots, n\}$ and $t > 0$, if 
$\nu(k, t)$ is not a leaf, then
$\left|\nu(k, t)\right| \leq 2v_kt$
($D_k^t$ has radius $v_kt$ and covers $\nu(k, t)$); and
\item for any $k \in \{1, \dots, n\}$ and $t \geq 0$, we have
$\left|\nu(k, t)\right| \geq v_kt/2$
($D_k^t$ has radius $v_kt$ and does not cover the parent of $\nu(k, t)$,
which has diameter $2\left|\nu(k, t)\right|$).
\end{enumerate}
\begin{figure}
\centering
\includegraphics{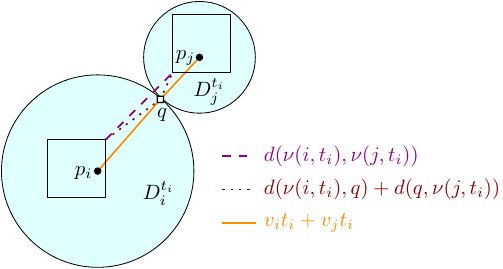}
\caption{When two disks meet, the two cells that
represent them are close.}
\label{fig:closesquares}
\end{figure}
\begin{table}
\centering
\includegraphics{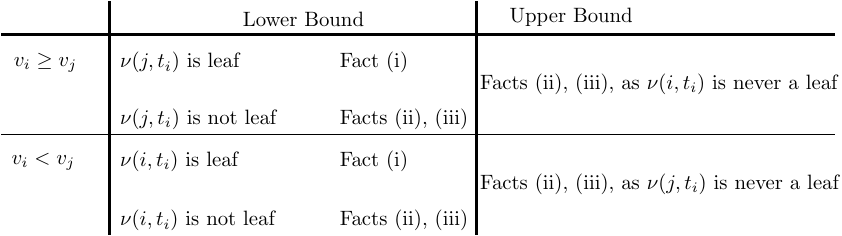}
\caption{Schematic overview for bounding the ratio of the 
  cell sizes.}
\label{fig:size_schematic}
\end{table}
For the first claim, on the distance between the cells,
let $q = \partial D_i^{t_i} \cap \partial D_j^{t_i}$ be the point
where $D_i$ and $D_j$ touch; see Figure~\ref{fig:closesquares}. 
By fact (iii), we have
$v_it_i \leq 2\left|\nu(i, t_i)\right|$
and
$v_jt_i \leq 2\left|\nu(j, t_i)\right|$.
Hence, it follows that
\[
d\left(\nu(i,t_i),\nu(j, t_i)\right)
\leq d\left(q, \nu(i,t_i)\right)+d\left(q, \nu(j, t_i)\right)
\leq v_it_i+v_jt_i
\leq 2\left|\nu(i, t_i)\right|+2\left|\nu(j, t_i)\right|,
\]
as claimed.  
Now we prove the second claim, on the ratio of the cell sizes;
see Table~\ref{fig:size_schematic} for an overview of the case 
distinction.  Suppose first that $v_i \geq v_j$.
If $\nu(j, t_i)$ is
a leaf, we have
$\left|\nu(i, t_i)\right|/
\left|\nu(j, t_i)\right| \geq 1$,
by fact (i).
If $\nu(j, t_i)$ is
not a leaf, it follows from
facts (ii) and (iii) that
\[
\frac{\left|\nu(i, t_i)\right|}
{\left|\nu(j, t_i)\right|}\geq \frac{v_it_i/2}{2v_jt_i}
\geq \frac{1}{4} \geq \frac{1}{4\Delta}.
  \]
By construction, the leaf cell that contains 
$p_i$ has diameter smaller than a quarter of
the smallest distance between disk centers.
Hence, 
the node $\nu(i, t_i)$ is not a leaf.
Thus, by facts (ii) and (iii), 
\[
\frac{\left|\nu(i, t_i)\right|}
{\left|\nu(j, t_i)\right|}\leq \frac{2v_it_i }{v_jt_i/2}
\leq 4 \cdot \frac{\max_i v_i}{\min_j v_j} = 4\Delta.
  \]
The argument for $v_j > v_i$ is analogous: 
if $\nu(i, t_i)$ is a leaf, then 
$\left|\nu(j, t_i)\right|/
\left|\nu(i, t_i)\right| \geq 1$, by fact (i).
If not, 
\[
\frac{\left|\nu(j, t_i)\right|}
{\left|\nu(i, t_i)\right|}\geq \frac{v_jt_i/2}{2v_it_i}
> \frac{1}{4} \geq \frac{1}{4\Delta},
  \]
by facts (ii) and (iii). Now, the node $\nu(j, t_i)$ cannot be a leaf,
so by facts (ii) and (iii)
\[
\frac{\left|\nu(j, t_i)\right|}
{\left|\nu(i, t_i)\right|}\leq \frac{2v_jt_i }{v_it_i/2}
\leq 4 \cdot \frac{\max_j v_j}{\min_i v_i} = 4\Delta.
  \]
The lemma follows.
\end{proof}

Lemma~\ref{lem:quad.def} gives a necessary condition
for the event that two quadtree cells could lead to
an elimination event. This allows us to focus on a
limited number of candidate pairs.
More formally, we say that two unrelated\footnote{That is,
no node is an ancestor or descendant of the other node.} nodes 
$\nu, \mu \in \quadtree$ form a \emph{candidate pair} 
if 
\[
d(\nu,\mu) \leq 2\left(|\nu|+|\mu|\right) \tag{*}
\]
and 
\[
\frac{|\nu|}{4\Delta} \leq |\mu|\leq 4\Delta|\nu|. \tag{**}
\]
We say that $\nu$ 
\emph{forms the candidate pair} $(\nu, \mu)$ with $\mu$.
We denote by $\cnp(\nu)$ the set of candidate pairs formed by
$\nu$. 
By Lemma~\ref{lem:quad.def}, to determine the elimination order, 
we need to consider only candidate pairs.
The following lemma uses a simple volume argument
to bound their number.
\begin{lemma}\label{lem:count_candidates}
Let $\nu \in \quadtree$.
Then, $\cnp(\nu)$ contains  $O(\alpha)$
candidate pairs $(\nu, \mu)$ with $|\nu| \leq |\mu|$.
The total number of candidate pairs is  $O(n \alpha \log \Phi)$.
\end{lemma}

\begin{proof}
Fix a node $\nu \in \quadtree$.
We claim that in each level of $\quadtree$, there are
at most $O(1)$ nodes $\mu$ with
$|\nu| \leq |\mu|$ such that $(\nu, \mu)$ is a candidate pair. 
This follows from a simple volume argument: fix a level $i$ of $\quadtree$, 
and let $X$ be the set of all nodes $\mu$ at level $i$ with 
$|\nu| \leq |\mu|$ and $(\nu, \mu) \in \cnp(\nu)$.
By (*) and our assumption $|\nu| \leq |\mu|$, we have  for all
$\mu \in X$ that
\[
d(\nu, \mu) \leq 2(|\nu| + |\mu|) \leq 4|\mu|.
\]
Hence, all cells $\mu \in X$ lie in
a region of diameter at most $9|\mu|$.
Since these cells have pairwise disjoint interiors and the same
size $|\mu|$, it follows that $|X| = O(1)$, as claimed.

Furthermore, by the definition of $\Phi$ and by (**), we have
$|\mu| =O(\min\{\Phi,\Delta\} )|\nu|$.
This implies that the levels of $\nu$ and $\mu$ in $\quadtree$ differ 
by $O(\alpha)$.
Hence, $\cnp(\nu)$ contains  $O(\alpha)$
candidate pairs $(\nu, \mu)$ with $|\nu| \leq |\mu|$.
Since $\quadtree$ has 
$O(n\log\Phi)$  nodes, and since the symmetry of (*) and
(**) shows that $(\nu, \mu) \in \cnp(\nu)$
if and only if $(\mu, \nu) \in \cnp(\mu)$, it follows that there are
$O(n\alpha\log\Phi)$ candidate pairs overall.
\end{proof}

Now we can describe our algorithm. As in Algorithm~\ref{alg:first},
we handle the disks in decreasing order of priority. 
Recall that by (\ref{equ:quadratic}), for each disk $D_i$, 
we need to find the disk $D_j$, $1 \leq j < i$, that first meets 
$D_i$ and is still alive. By Lemma~\ref{lem:quad.def}, it suffices
to focus on disks that correspond to candidate pairs.
Throughout the algorithm, 
we compute for each node $\mu$ in $\quadtree$ the
index $D(\mu)$ of the (at most one) disk that occupies it at some point 
in time.
When processing a disk $D_i$, we
start from the leaf node for $D_i$, and we 
simulate the growth of $D_i$ by tracing a path to the root.
As we follow the leaf-root path for $D_i$, we check all the
candidate pairs $(\nu, \mu) \in \cnp(\nu)$ for the current node $\nu$, 
to see if a node $\mu$ is occupied by a disk that could eliminate
$D_i$ (i.e., a disk of higher priority that is still alive when it meets
$D_i$). These disks have been computed in previous iterations
of the algorithm.
We continue along the leaf-root path for $D_i$ until it is clear
that $D_i$ has been eliminated before it can occupy the current node
$\nu$, setting the $D(\nu)$ variables accordingly.
The pseudocode in Algorithm~\ref{alg:Quadtree} provides the details.
Initially, we set  all $D(\nu) = \perp$.
We use $\tau(\nu, i)$ to denote the first time at 
which $b(\nu)$ is covered by the disk $D_i$. 
We will describe below how the sets $\cnp(\nu)$ can be computed efficiently.

\begin{algorithm}
  \caption{A simple quadtree-based algorithm\label{alg:Quadtree}}
  \begin{algorithmic}[1]
    \Function{EliminationOrder}{$p_1,\dots,p_n$, $v_1,\dots,v_n$}
      \State $\quadtree \gets$ ConstructQuadTree($p_1,\dots,p_n$)
      \label{line:build_qt}
      \State $D(\nu) \gets \perp$ for every node $\nu$ of $\quadtree$

      \For{$i \gets 1, \ldots ,n$} 
        \State $\nu \gets$ getLeaf($p_i$) 
        \State $t \gets \infty$ 
	\label{line:init_ti}
        \While{$\nu \neq \perp$ and $t \geq \tau(\nu,i)$}
	  \label{line:while}
          \State $D(\nu) \gets i$ 
	  \For{$(\nu, \mu)$ in $\cnp(\nu)$} 
	  \If{$D(\mu) \neq \perp$ and $t_{D(\mu)} \geq t(i,D(\mu))$}
	      \label{line:test_collision}
	      \State $t \gets \min\Big(t, t\big(i,D(\mu)\big)\Big)$
             \EndIf
          \EndFor
          \State $\nu \gets p(\nu)$
        \EndWhile
        \State $t_i \gets t$
      \EndFor
      \State $S \gets (D_1,\dots,D_n)$ \State Sort $S$ using key $t_i$
         for each disk $D_i$ \State \Return $S$
	 \label{line:sorting}
    \EndFunction
  \end{algorithmic}
\end{algorithm}

\begin{theorem}\label{thm:quadtree}
  The elimination sequence of $n$ growing disks can be computed in
  $O(n\alpha\log\Phi)$ time and $O(n \log \Phi)$ space, 
  where $\alpha=\min\{\log\Phi, \log\Delta\}$.
\end{theorem}
\begin{proof}
The outer \textbf{for}-loop iterates over the input disks by decreasing
order of priority. In the \textbf{while}-loop, the algorithm
traverses each node $\nu \in \quadtree$ from the
leaf-node with $p_i$ to the root. It updates
$D(\nu)$, until it encounters a node $\nu$ with
$t < \tau(\nu,i)$. The inner \textbf{for}-loop iterates over
every candidate pair $(\nu, \mu)$ in $\cnp(\nu)$.
It checks if disk $D_i$ and $D(\mu)$
might touch by computing the time $t(i, D(\mu))$;
if so, it updates the tentative elimination time for $D_i$.
To show correctness, we prove that the algorithm maintains the
following invariant: after $i$ iterations of the \textbf{for}-loop,
the algorithm has correctly computed the elimination times 
$t_1, \dots, t_i$ for $D_1, \dots, D_i$. 
Furthermore, for each node $\nu$ of $\quadtree$, we have
$D(\nu) \in \{1, \dots, i, \perp\}$, and if there is a point 
in time when $\nu$ is occupied by the disk $D_j$, $1 \leq j \leq i$, 
then $D(\nu) = j$. 

The invariant holds after the first 
iteration. This is because $t$ is set to $\infty$
in Line~\ref{line:init_ti}, and all $D(\mu)$'s are initalized to
$\perp$, so that the \textbf{while}-loop will proceed all the way
to the root and set all $D(\nu)$'s along the leaf-root path
for $p_1$ to $1$.
Now suppose that $i \geq 2$ and 
that $D_i$ is eliminated by the disk $D_j$, with $1 \leq j < i$.
By Lemma~\ref{lem:quad.def}, the pair
$(\nu, \mu) = (\nu(i, t_i), \nu(j, t_i))$ is a candidate pair
in $\cnp(\nu)$. Furthermore,
$\nu$ lies on the leaf-root path for $p_i$ and by the
inductive hypothesis, we have $D(\mu) = j$.
Again by the inductive hypothesis and by the test 
$t_{D(\mu)} \geq t(i, D(\mu))$ in Line~\ref{line:test_collision}, 
we have $t \geq t_i$ throughout the \textbf{while}-loop.
Thus, the \textbf{while}-loop will visit $\nu$ and 
the candidate pair $(\nu, \mu)$ will be considered.
At this point, Algorithm~\ref{alg:Quadtree} will detect
the elimination event and set $t$ to $t_i$.
After that, the test $t \geq \tau(\nu, i)$ in Line~\ref{line:while} 
ensures that the remaining variables $D(\nu)$ are set correctly.
In particular, the algorithm does not overwrite any
other such values from previous iterations. Thus, the invariant is 
maintained.

We now turn to the running time.
We can compute $\quadtree$ in Line~\ref{line:build_qt} 
in $O(n\log\Phi)$ time and space.
Furthermore, by Lemma~\ref{lem:count_candidates}, there are
$O(n \alpha \log \Phi)$ candidate pairs overall, so that the total time
for the \textbf{for}-loop, excluding the time for finding the candidate
pairs, is $O(n \alpha \log \Phi)$. We find the candidate pairs
without additional asymptotic overhead as follows:
to determine the leaf for $p_i$, we walk down from the root,
going from one node that contains $p_i$ to the next. 
Suppose that we are currently at level $\ell$, and that
the node $\nu$ contains $p_i$ in $b(\nu)$. We find all nodes $\mu$ at
level $\ell$ with $d(\nu, \mu) \leq 2(|\nu| + |\mu|)$, and we
store them in a list for level $\ell$. For this, we can use
appropriate pointers in $\quadtree$, or we perform an appropriate
root-leaf traversal of $\quadtree$ that keeps track of all nodes 
close to the nodes containing $p_i$. By a simple volume argument, 
there are $O(1)$ such nodes at each level, for a total of $O(\log \Phi)$.
To find $\cnp(\nu)$ for a node $\nu$, we consider the list for the
parent of $\nu$ that is $O(\alpha)$ levels above $\nu$ (including the parent),
and we perform a depth-first search in $\quadtree$ 
of these nodes and their descendants for all candidate pairs.
This takes $O(\alpha + |\cnp(\nu)|)$ time,
and it can be implemented so that the space 
requirement does not exceed $O(n \log \Phi)$, since we can process the
candidate pairs immediately as we discover them.
Finally, since $\Phi = \Omega(\sqrt{n})$, the sorting step in 
Line~\ref{line:sorting}
does not increase the asymptotic running time or space.\footnote{By a packing 
argument, 
the spread of any $d$-dimensional $n$-point set is 
$\Omega(n^{1/d})$: if any two points have distance at least $1$,
the point set must cover at least $\Omega(n)$ units of
volume and hence must have diameter $\Omega(n^{1/d})$.}
\end{proof}

\subsection{Using a compressed quadtree}\label{sec:comp}
  
We now speed up
Algorithm~\ref{alg:Quadtree} with the help of 
a \emph{compressed quadtree} $\cquadtree$. In $\cquadtree$, 
the number of nodes is reduced to $O(n)$ by replacing certain long paths 
in the uncompressed quadtree $\quadtree$ with single \emph{compressed} 
edges. Now, the definition of the candidate pairs 
becomes more tricky, because the cells 
identified in Lemma~\ref{lem:quad.def} might no longer 
be present after the compression.
Thus, we need a way to map candidate pairs in
$\quadtree$ to candidate pairs in $\cquadtree$. 
It is important that no elimination event
is missed and that the number of candidate pairs
is small. To achieve this we project a 
candidate pair $(\nu, \mu)$ in $\quadtree$ to the pair
given by the lowest ancestors of $\nu$ and of $\mu$ that 
appear in $\cquadtree$. With these
\emph{compressed candidate pairs}, we can essentially
run Algorithm~\ref{alg:Quadtree} on $\cquadtree$,
with minor modifications. To complete our
result, we must show that the compressed candidate pairs
are few and can be found efficiently, and
that no elimination event is missed.
Details follow.

We begin with the formal definition of the compressed
quadtree.
Let $\quadtree$ be the (uncompressed) quadtree for
the $n$ disk centers, as in Section~\ref{sec:uncomp}.
We describe how to obtain the compressed quadtree $\cquadtree$
from  $\quadtree$.
A node $\nu$ in $\quadtree$ is 
\emph{empty} if $b(\nu)$ does not contain a disk-center,
and \emph{non-empty} otherwise. A \emph{singular path} $\sigma$ in 
$\quadtree$ is a path $\nu_1, \nu_2, \dots, \nu_k$ of nodes such
that (i) $\nu_k$ is a non-empty leaf or has at 
least two non-empty children; and
(ii) for $i = 1, \dots, k-1$, the node $\nu_{i+1}$
is the only non-empty child of $\nu_i$.
We call $\sigma$ \emph{maximal} if it cannot be extended
by the parent of $\nu_1$ (either because $\nu_1$ is the
root or because $p(\nu_1)$ has two non-empty children).
For each maximal singular path 
$\sigma = \nu_1, \dots, \nu_k$ in $\quadtree$, we remove from $\quadtree$
all proper descendants of $\nu_1$  that are not descendants of $\nu_k$,
together with their incident edges. Then, we add a new 
\emph{compressed edge} between $\nu_1$ and $\nu_k$.
The resulting tree $\cquadtree$ has $O(n)$ nodes. Each internal node has
one or four children.\footnote{According to our definition, 
the compressed quadtree may contain empty leaves, namely 
empty leaves that are children of nodes with at least two non-empty 
children. This empty leaves do not belong to any 
singular path.}
See Figure~\ref{fig:cquadtree} (right) for an illustration.
There are algorithms to compute $\cquadtree$ in 
$O(n\log n)$ time and $O(n)$ space (see, e.g. Har-Peled's
book~\cite[Theorem~2.9]{Sariel}
or Buchin~\etal for a version that does not need the
floor function~\cite[Appendix~A]{BuchinLoMoMu11}). 

We now describe how to map the candidate pairs from $\quadtree$
to candidate pairs in $\cquadtree$.  A node $\nu$ from
$\quadtree$ may appear as a node in $\cquadtree$ or not. We 
let $\pi(\nu)$ be the lowest ancestor node 
 of $\nu$ (including $\nu$) in $\quadtree$ that appears also in
$\cquadtree$. We call $\pi(\nu)$ the \emph{upward projection} 
of $\nu$ in $\cquadtree$.
For a node $\nu$ that appears in $\cquadtree$,
we denote by $\Pi(\nu) = \{ \nu' \in \quadtree \mid \pi(\nu') = \nu\}$
the set of all nodes in $\quadtree$ that project to $\nu$. If
$\nu$ is a leaf or has four children in $\cquadtree$, then
$\Pi(\nu) = \{\nu\}$. Otherwise, if $\nu$ has one child
in $\cquadtree$, then $\Pi(\nu)$ contains the nodes of the
maximal singular path that starts in $\nu$, without the last
vertex. Furthermore, for $\nu \in \cquadtree$, we write 
$p_C(\nu)$ the parent of $\nu$ in $\cquadtree$ (we set
$p_C(\nu) = \perp$, if $\nu$ is the root).
We define the set of \emph{compressed candidate pairs}
$\cnp_C(\nu)$ for $\nu$ in $\cquadtree$ as 
\[
  \cnp_C(\nu) = \big\{(\pi(\nu'), \pi(\mu)) \mid 
  \nu' \in \Pi(\nu) \text{ and } (\nu', \mu) \in \cnp(\nu') 
  \text{ and } \pi(\nu') \neq \pi(\mu)
  \big\}.
  \tag{***}
\]
In other words,  we obtain the compressed candidate pairs
for a node $\nu$ in $\cquadtree$ by taking the upward
projections of all candidate pairs in $\quadtree$ 
where the first component projects to $\nu$.

We can now describe our modified algorithm, 
Algorithm~\ref{alg:compressed}. Essentially, 
Algorithm~\ref{alg:compressed} works in the same 
way as Algorithm~\ref{alg:Quadtree}: we go through 
the disks by decreasing order of priority, and for 
each disk $D_i$, we simulate the growth process of $D_i$
by walking along the leaf-root path for the disk 
center $p_i$ in the compressed quadtree $\cquadtree$,
while checking for elimination events with nearby 
disks of higher priority. There are two 
differences: first, instead of the  candidate
pairs $\cnp(\nu)$ of the current node $\nu$, we now use the 
compressed candidate pairs $\cnp_C(\nu)$ to check for 
elimination events; second, the termination condition of 
the \textbf{while}-loop (line~\ref{line:while_compressed})
is modified: instead of comparing $t$ with the 
time $\tau(\nu, i)$ when $D_i$ first covers $b(\nu)$,
we use $\tau_C(\nu, i)$, the time when 
$D_i$ first covers the box for any node in $\Pi(\nu)$. 
This ensures that we do not miss an elimination event on a 
singular path starting at $\nu$.
Next, we argue that Algorithm~\ref{alg:compressed} correctly 
computes the elimination sequence. Then, we will discuss 
an efficient implementation. The correctness for
Algorithm~\ref{alg:compressed} follows from essentially the
same argument as for Algorithm~\ref{alg:Quadtree}. 

\begin{algorithm}
  \caption{Using the Compressed Quadtree\label{alg:compressed}}
  \begin{algorithmic}[1]
    \Function{EliminationOrder}{$p_1,\dots,p_n$, $v_1,\dots,v_n$}
      \State $\cquadtree \gets$ ConstructCompressedQuadTree($p_1,\dots,p_n$)
      \label{line:build_cqt}
      \State $D_C(\nu) \gets \perp$ for every node $\nu$ of $\quadtree$

      \For{$i \gets 1, \ldots ,n$} 
        \State $\nu \gets$ getCompressedLeaf($p_i$) 
        \State $t \gets \infty$ 
	\label{line:init_ti_compressed}
        \While{$\nu \neq \perp$ and $t \geq \tau_C(\nu,i)$}
        \label{line:while_compressed}
          \State $D_C(\nu) \gets i$ 
	  \For{$(\nu, \mu)$ in $\cnp_C(\nu)$} 
	  \label{line:cnpc}
	  \If{$D_C(\mu) \neq \perp$ and $t_{D_C(\mu)} \geq t(i, D_C(\mu))$}
	      \label{line:test_collision_compressed}
	      \State $t \gets \min\Big(t, t\big(i,D_C(\mu)\big)\Big)$
             \EndIf
          \EndFor
          \State $\nu \gets p_C(\nu)$
        \EndWhile
        \State $t_i \gets t$
      \EndFor
      \State $S \gets (D_1,\dots,D_n)$ \State Sort $S$ using key $t_i$
         for each disk $D_i$ \State \Return $S$
	 \label{line:sorting_compressed}
    \EndFunction
  \end{algorithmic}
\end{algorithm}

\begin{lemma}\label{lem:compressed_correctness}
Algorithm~\ref{alg:compressed} correctly computes the elimination 
sequence.
\end{lemma}

\begin{proof}
We prove that Algorithm~\ref{alg:compressed} maintains the following 
invariant: after $i$ iterations of the \textbf{for}-loop,
we have the correct elimination 
times $t_1, \dots, t_i$ for $D_1, \dots, D_i$. Furthermore,
for each node $\nu$ of $\cquadtree$, we have
$D_C(\nu) \in \{1, \dots, i, \perp\}$, and
if there is a point in time when a node in $\Pi(\nu)$ is
covered by the disk $D_j$, $1 \leq j \leq i$, then
$D_C(\nu) = j$ (this is well defined, because the nodes
along a singular path can be occupied by at most one
disk).

The invariant holds for $i = 1$, because $t$ is set to 
$\infty$ in Line~\ref{line:init_ti_compressed} and for all the 
nodes $\nu$ in $\cquadtree$, the value $D_C(\nu)$ is initialized 
to $\perp$, so that no elimination event will be detected in 
Line~\ref{line:test_collision_compressed}. The values $D_C(\nu)$ 
on the leaf-root path of $p_1$ are all set to $1$, as desired.
Now, suppose that $i \geq 2$ and that $D_i$ is eliminated
by the disk $D_j$, with $1 \leq j < i$.
By Lemma~\ref{lem:quad.def}, the pair
$(\nu, \mu) = (\nu(i, t_i), \nu(j, t_i))$ is a candidate pair
in $\cnp(\nu)$. Furthermore,
$\nu$ lies on the leaf-root path for $p_i$ in $\quadtree$.
By~(***), it follows that $(\pi(\nu), \pi(\mu))$
is a compressed candidate pair in $\cnp(\pi(\nu))$ 
and that $\pi(\nu)$ lies on the leaf-root
path for $p_i$ in $\cquadtree$.
By the inductive hypothesis, we have $D_C(\pi(\mu)) = j$.
Again by the inductive hypothesis and by the test 
$t_{D_C(\mu)} \geq t(i, D_C(\mu))$ in 
Line~\ref{line:test_collision_compressed}, 
we have $t \geq t_i$ throughout the \textbf{while}-loop.
Thus, the \textbf{while}-loop will visit $\pi(\nu)$ and 
consider the compressed candidate pair $(\pi(\nu), \pi(\mu))$.
At this point, Algorithm~\ref{alg:compressed} will detect
the elimination event and set $t$ to $t_i$.
After that, the test $t \geq \tau_C(\nu, i)$ in 
Line~\ref{line:while_compressed} 
ensures that the remaining variables $D_C(\nu)$ are set correctly.
In particular, the algorithm does not overwrite any
other such values from previous iterations. Thus, the invariant is 
maintained and the correctness of Algorithm~\ref{alg:compressed} follows.
\end{proof}

\begin{figure}
\centering
\includegraphics[scale=1.2]{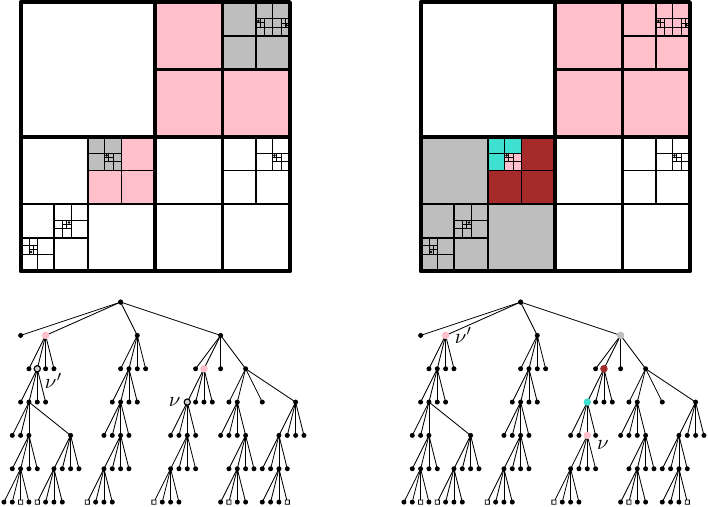}
\caption{(left) If two nodes $\nu$ and $\nu'$ form a candidate pair,
so do their parents. (right) If $|\nu| \leq |\nu'|$ and 
$\nu$ forms a candidate pair with $\nu'$, then so do all
ancestors of $\nu$ up to the level of $\nu'$.
}
\label{fig:upward_closure}
\end{figure}
Next, we discuss how to implement Algorithm~\ref{alg:compressed}
efficiently. First we must bound the number of compressed 
candidate pairs.
For this, we need two closure properties for candidate pairs
in $\quadtree$: if two nodes in
$\quadtree$ form a candidate pair, then so
do their parents; and if two nodes $(\nu, \mu)$ form a 
candidate pair in which $\nu$ lies at a lower level than $\mu$,
then all the ancestors of $\nu$ up to the level of $\mu$ also
form candidate pairs with $\mu$. See Figure~\ref{fig:upward_closure}
for an illustration.
\begin{lemma}\label{lem:cnp.parent}
Let $\nu, \mu$ be nodes of $\quadtree$ such that 
$(\nu, \mu) \in \cnp(\nu)$ and $p(\nu) \neq p(\mu)$.
Then, we have
\begin{enumerate}[label=(\roman*)]
\item $(p(\nu), p(\mu)) \in \cnp (p(\nu))$, i.e., the parents
of $\nu$ and $\mu$ also form a candidate pair; and
\item if $|\nu| \le |\mu|$, then 
$(\nu', \mu) \in \cnp(\nu')$
for any ancestor $\nu'$ of $\nu$ with $|\nu'| \le |\mu|$,
i.e., all the ancestors of $\nu$ up to the level of $\mu$ also form
candidate pairs with $\mu$.
\end{enumerate}
\end{lemma}

\begin{proof}
We must check properties (*) and (**) of a 
candidate pair.
For~(i), property (*) holds because
\[
d(p(\nu), p(\mu)) \leq d(\nu, \mu) \leq 2(|\nu|+|\mu|) 
\leq 2(|p(\nu)|+p(|\mu|)),
\] 
where in the first and third inequality we used the fact that
$b(\nu) \subset b(p(\nu))$ and $b(\mu) \subset b(p(\mu))$ and
in the second inequality we used that $(\nu, \mu)$ is a 
candidate pair. Property (**) holds because
\[
\frac{|p(\mu)|}{|p(\nu)|} = 
\frac{|\mu|/2}{|\nu|/2} =
\frac{|\mu|}{|\nu|} \in \left[\frac{1}{4\Delta}, 4\Delta\right],
\]
since $(\nu, \mu)$ is a candidate pair.
The argument for (ii) is analogous. Since $b(\nu)$ is a subset
of $b(\nu')$ and since $(\nu, \mu)$ is a candidate pair, we have 
\[
d(\nu', \mu)
\leq d(\nu, \mu)
\leq 2(|\nu|+|\mu|) 
\leq 2(|\nu'|+|\mu|),
\]
so property (*) holds. Furthermore, since $|\nu| \leq |\nu'| \leq |\mu|$
and since $(\nu, \mu)$ is a candidate pair, we have
\[
\frac{1}{4\Delta} \leq
1 \leq \frac{|\mu|}{|\nu'|} \leq \frac{|\mu|}{|\nu|} \leq 4\Delta,
\]
which shows property (**).
\end{proof}

\begin{figure}
\centering
\includegraphics[scale=1.2]{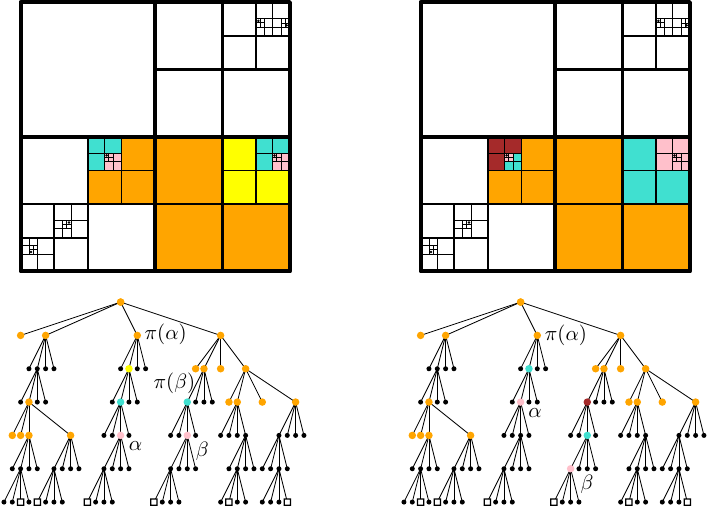}
\caption{Compressed candidate pairs
can be charged to regular candidate pairs
in $\cquadtree$ where the second node is larger: 
apply 
Lemma~\ref{lem:cnp.parent}(i)  until
reaching a node in $\cquadtree$;
(left) If this is the smaller node, 
the projection of the other node stays the same;
(right) if it is the larger node, we
apply Lemma~\ref{lem:cnp.parent}(ii).
}
\label{fig:projection_lemma}
\end{figure}
The next lemma provides a way to charge compressed candidate pairs
in $\cquadtree$ to candidate pairs in $\quadtree$.
More precisely, we show that each compressed candidate pair
$(\nu, \mu)$ in $\cquadtree$ can be obtained by taking
a candidate pair $(\sigma, \tau) \in \cnp(\nu) \cup \cnp(\mu)$ with
$|\sigma| \leq |\tau|$, and by projecting $\tau$ upwards;
see Figure~\ref{fig:projection_lemma}.

\begin{lemma} \label{lem:charging}
Let $\nu, \mu$ be two nodes in $\cquadtree$ 
with $(\nu, \mu) \in \cnp_C(\nu)$. Then, there
are two nodes $\sigma, \tau$ in $\quadtree$ 
such that (i) $(\sigma, \tau) \in \cnp(\sigma)$, i.e., $\sigma$ and $\tau$
form a candidate pair in $\quadtree$;
(ii) $|\sigma| \leq |\tau|$, i.e., $\sigma$ is not larger than
$\tau$; and (iii) $(\nu, \mu) = (\sigma, \pi(\tau))$ 
or $(\nu, \mu) = (\pi(\tau), \sigma)$, i.e., the compressed candidate
pair $(\nu, \mu)$ is obtained by taking a candidate pair for
$\nu$ or for  $\mu$ and by projecting the other component upwards.
\end{lemma}

\begin{proof}
Since $(\nu, \mu) \in \cnp_C(\nu)$, definition
(***) implies that there are two nodes
$\alpha$, $\beta$ in $\quadtree$ where
(i) $\alpha$ and $\beta$ form a candidate pair, i.e.,
$(\alpha, \beta) \in \cnp(\alpha)$;
(ii) the notation is
such that $\alpha$ is not larger than $\beta$, i.e.,
$|\alpha| \leq |\beta|$; 
and
(iii) the upward projections of $\alpha$ and $\beta$ 
(in the right order) give the compressed candidate pair $(\nu, \mu)$,
i.e.,
$\{\nu, \mu\} = \{\pi(\alpha), \pi(\beta)\}$.
We repeatedly apply 
Lemma~\ref{lem:cnp.parent}(i) to $(\alpha, \beta)$, until we
meet $\pi(\alpha)$ or $\pi(\beta)$, whichever happens first. By assumption,
$\pi(\alpha)$ and $\pi(\beta)$ are distinct, so all the parents
along the way are also distinct, and Lemma~\ref{lem:cnp.parent}(i)
is applicable.
See Figure~\ref{fig:projection_lemma} for an illustration.

Suppose we meet $\pi(\alpha)$ first.
In this case, we set $\sigma = \pi(\alpha)$.
Then, it follows that $(\sigma, \tau) \in \cnp(\sigma)$ 
for some ancestor $\tau$ of $\beta$ in $\quadtree$ with 
$|\sigma| \leq |\tau|$. Since $\sigma = \pi(\alpha)$ is 
encountered first, the upward projections of
$\tau$ and $\beta$ in $\cquadtree$ are the same, i.e.,
$\pi(\tau) = \pi(\beta)$. Hence, the pair $(\sigma, \tau)$
has all the desired properties.

Second, suppose we meet $\pi(\beta)$ first.
Consider the ancestor $\alpha'$ of $\alpha$ that has
the same size as $\pi(\beta)$.
If $\pi(\alpha)$ appears on the path from $\alpha$ to $\alpha'$
in $\quadtree$, we set $\sigma = \pi(\alpha)$ and $\tau = \pi(\beta)$. 
Otherwise,
we set $\sigma = \pi(\beta)$ and $\tau = \alpha'$.
In either case, $|\sigma| \leq |\tau|$. Furthermore,
by Lemma~\ref{lem:cnp.parent}(ii), we have
$(\sigma, \tau) \in \cnp(\sigma)$. Finally, the upward
projections are maintained.
Hence, $(\sigma, \tau)$ again has all the desired 
properties.
\end{proof}

Now, we can bound the number of compressed candidate pairs with
a simple charging argument.
\begin{lemma}\label{lem:bound_cands}
The total number
of compressed candidate pairs is $O(n\alpha)$.
\end{lemma}

\begin{proof}
Let $(\nu, \mu)$ be a compressed candidate pair.
We use Lemma~\ref{lem:charging} to charge
$(\nu, \mu)$ to a pair 
$(\sigma, \tau) \in \cnp(\nu) \cup \cnp(\mu)$
with $|\sigma| \leq |\tau|$ and $\{\pi(\sigma), \pi(\tau)\} = \{\nu, \mu\}$.

Now, let $\nu$ be a node of $\cquadtree$. Every candidate pair
$(\nu, \mu) \in  \cnp(\nu)$ with $|\nu| \leq |\mu|$ is charged
at most twice, namely (potentially) by $(\nu, \pi(\mu))$ and by 
$(\pi(\mu), \nu)$. By Lemma~\ref{lem:count_candidates}, there are 
$O(\alpha)$ candidate pairs $(\nu, \mu) \in  \cnp(\nu)$ with 
$|\nu| \leq |\mu|$. 
Since $\cquadtree$ has $O(n)$ nodes, the claim follows.
\end{proof}

Now, we have enough tools to find all the compressed candidates
in Line~\ref{line:cnpc} efficiently.
\begin{lemma}\label{lem:count_candidate_C}
During Algorithm~\ref{alg:compressed},
we can enumerate all compressed candidates $\cnp_C(\nu)$ for the 
nodes $\nu$ visited by the algorithm in 
total time $O(n \alpha)$ and space $O(n)$.
\end{lemma}

\begin{proof}
During preprocessing, we compute for each node $\nu$ in $\cquadtree$
\emph{neighbor pointers} 
to all nodes $\mu \in \cquadtree$ where $\Pi(\mu)$ contains
a node $\mu'$ with $|\nu| = |\mu'|$ and $d(\nu, \mu') \leq 4|\nu|$.
For each $\nu$, there are $O(1)$ such pointers (a simple volume
argument), and they can be found
in $O(n)$ time and space by a top-down traversal.

Now suppose we want to enumerate the compressed candidate pairs for a
node $\nu$ in $\cquadtree$. Let $(\nu, \mu) \in \cnp(\nu)$ be such
a candidate. By Lemma~\ref{lem:charging}, we have either (i)
$(\nu, \mu) = (\nu, \pi(\mu'))$ for a candidate pair 
$(\nu, \mu') \in \cnp(\nu))$ with $|\nu| \leq |\mu'|$ or 
(ii) $(\nu, \mu) = (\pi(\nu'), \mu)$,
where $\mu$ appears in $\cquadtree$ and $(\mu, \nu')$ is
a candidate pair in $\quadtree$ with $|\mu| \leq |\nu'|$.

To find the compressed candidate pairs of type (i), we 
enumerate all (regular) candidate pairs for $\nu$, using 
the neighbor pointers for $\nu$ and its ancestors in a similar procedure as in 
Theorem~\ref{thm:quadtree}.
By Lemma~\ref{lem:count_candidates}, 
this takes $O(\alpha)$ time and no additional
space if we process the compressed candidate pairs immediately without 
storing them. For compressed candidate pairs of type (ii), we 
must enumerate all 
nodes $\mu \in \cquadtree$
such that $\Pi(\nu)$ contains
a node $\nu'$ with $(\mu, \nu') \in \cnp(\mu)$ and $|\mu| \leq |\nu'|$.
The crucial observation is that by Lemma~\ref{lem:cnp.parent},
these  nodes form a connected subtree under each neighbor 
node of $\nu$. Thus, we can find them by following the neighbor
pointers for $\nu$ and by traversing each such subtree as long
as a compressed candidate pair is found. We can check whether
a node $\mu \in \cquadtree$ forms a compressed candidate pair with
$\nu$ in $O(1)$ by elementary calculations involving
the floor function.\footnote{We can also do without the floor function
if we slightly relax the notion of a compressed candidate pair.}
Thus, the time is proportional to the number of
distinct compressed candidate pairs that are discovered. No additional space
is necessary, because the compressed candidate pairs can be
processed immediately. Since each node in $\cquadtree$
is visited at most once by Algorithm~\ref{alg:compressed}, the
result now follows from Lemma~\ref{lem:bound_cands}.
\end{proof}

The following theorem summarizes our result for this section.
\begin{theorem} The elimination sequence of $n$
disks can be computed in $O(n\log n + n \alpha)$ time 
and $O(n)$ space, where $\alpha=\min\{\log\Phi, \log\Delta\}$.
\end{theorem}

\begin{proof}
By Lemma~\ref{lem:compressed_correctness}, Algorithm~\ref{alg:compressed}
correctly computes the elimination sequence. The compressed quadtree
$\cquadtree$ can be constructed in $O(n \log n)$ time and $O(n)$
space. This is also the time needed for the final sorting step.
Since Algorithm~\ref{alg:compressed} visits each node of $\cquadtree$
at most once, and since $\cquadtree$ has $O(n)$ nodes, the time
for the \textbf{for}-loop (without the time for computing the
compressed candidates) is $O(n \alpha)$, by Lemma~\ref{lem:bound_cands}.
It uses no additional space.
Finally, by Lemma~\ref{lem:count_candidate_C}, the additional time
for enumerating the compressed candidate pairs is $O(n \alpha)$,
using $O(n)$ space. The result follows.
\end{proof}

\section{Lower bound}

To complement our results, we provide a lower bound for finding
elimination sequences in the algebraic decision tree model.
Formann~\cite{Formann93} argued that the weighted-closest pair
problem has an $\Omega(n \log n)$ lower bound in the
algebraic decision tree model, by a reduction from
the closest-pair-problem. Since finding the elimination order
is a more general problem, this also implies an $\Omega(n \log n)$
lower bound for our problem. Here, we provide a slightly
stronger result by showing that the \emph{sorting problem}
reduces to finding elimination orders.

\begin{theorem}\label{thm:lowerbound}
There is a reduction from the sorting problem to the
elimination order problem.
In particular, it takes at least $\Omega(n \log n)$ time to 
find the elimination order of a set of $n$ 
growing disks or squares in the plane under the algebraic decision tree model. 
\end{theorem}

\begin{figure}
\centering
\includegraphics{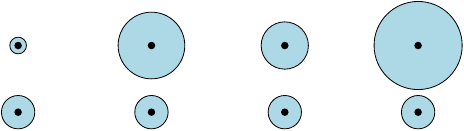}
\caption{The lower bound reduction: the disk centers
are arranged in two rows, where two consecutive
disk centers have distance $2$. The distance between two
rows is $1$ (actually, the points are slightly perturbed 
for general position; this is omitted in the figure). 
The disks centered at the bottom row grow
uniformly with rate $1$. The disks in the top row grow
at varying speeds in $(1,2)$. The elimination order of the disks
in the top row corresponds to the reverse sorted order of their
growth rates.}
\label{fig:lower_bound}
\end{figure}

\begin{proof}
We show that the problem of sorting $n$ numbers 
$v_{n+1}, \dots, v_{2n} \in (1, 2)$ can be reduced to finding the
elimination order of $2n$ disks in the plane.  
This implies an $\Omega(n \log n)$ lower bound in the algebraic
decision tree model.  

Our reduction proceeds as follows:
Suppose we are given $n$ numbers 
$v_{n+1}, \dots, v_{2n} \in (1, 2)$, to be sorted.
Set $\eps = 1/10n^3$.
We define $2n$ growing disks $D_1, \dots, D_{2n}$ as follows:
for $i = 1, \dots, n$, we center the disk $D_i$ at 
$p_i = \big(2i + i^2 \eps, 0\big)$ and give it the growth
rate $v_i = 1$.
For $i = n + 1, \dots, 2n$, we position the disk 
$D_i$ at 
$p_i = (2i + i^2\eps, 1)$ with growth rate $v_i$ as
in the input; see Figure~\ref{fig:lower_bound} for an
illustration.
Observe that disk $D_{n+i}$ will be eliminated by disk $D_i$ at 
time $t_{n+i}=t(n+i,i)=1/(1 + v_{n+i}) < 1/2$, since 
$t_i > 1/2$ for $1 \leq i \leq n$.
Then, the elimination order of $D_1, \dots, D_{2n}$
lets us deduce the reverse sorted order of 
$\{v_{n+1},\dots,v_{2n}\}$.
An analogous argument also applies to squares.
\end{proof}

We remark that Theorem~\ref{thm:lowerbound} also shows that 
in general the problem does not become easier if we are interested only 
in the elimination order and not the exact elimination times.

\section{Conclusion}

We have presented the first truly subquadratic algorithm for 
the problem of computing the elimination order and elimination times of
a sequence of $n$ growing disks in the plane. Our approach
is very general and also applies to other shapes. However, 
it still falls short of reaching a near-linear time algorithm,
except for the special case of growing cubes. Thus, the most pressing question
remains: can we compute the elimination order of $n$ growing disks
in the plane in $O(n \log n)$ time?

Our algorithm that uses compressed quadtrees comes close,
but it depends on additional parameters of the input: if 
the growth rates vary wildly, or if the points are arranged
unevenly, the running times may deteriorate. Perhaps a more
careful handling of these inputs could enable us to avoid
this dependence. It would also be interesting to see if and how 
the quadtree approach can be adapted to higher dimensions.

Finally, many further well-motivated variants of the
problem are possible. For example, 
Castermans~\etal~\cite{cssv-acgs-18} consider the setting
where two touching disks are replaced by a new, common,
disk, instead of one of them disappearing. It is a promising
research direction to explore these variants and to see
in how far our techniques are applicable or which new
ideas are required.

\subparagraph*{Acknowledgments.}
This work was initiated during the 20th Korean Workshop on 
Computational Geometry. The authors would like to thank the 
other participants for motivating and insightful discussions. 
We would also like to thank the anonymous reviewers for their close
reading of the paper and for many helpful comments that
significantly improved the presentation of the paper.

\bibliography{gdisks}

\end{document}